\documentclass[singlecolumn,12pt]{article}
\usepackage{amsmath,amssymb,amsthm}
\usepackage{algorithm}
\usepackage{times}
\usepackage[latin1]{inputenc}
\usepackage[T1]{fontenc}
\usepackage{latex8}
\usepackage{graphicx}
\usepackage{verbatim}
\usepackage{makeidx}
\usepackage{url}
\newtheorem{example}{Example}
\newtheorem{theorem}{Theorem}

\newtheorem{remark}{Remark}

\begin{document}

\title{Polynomial Circuit Verification using BDDs}

\author{
 {\centering{\begin{tabular}{c}
Rolf Drechsler  \\
\\
Institute of Computer Science \\
University of Bremen \\
28359 Bremen, Germany  \\
\multicolumn{1}{c}{drechsler@uni-bremen.de}
\end{tabular}
 }} }
\maketitle
\thispagestyle{empty}
\begin{abstract}
\begin{quote}
Verification is one of the central tasks during circuit design. While most of the approaches have exponential worst-case behaviour, in the following techniques are discussed for proving polynomial circuit verification based on Binary Decision Diagrams (BDDs). It is shown that for circuits with specific structural properties, like e.g.~tree-like circuits, and circuits based on multiplexers derived from BDDs complete formal verification can be carried out in polynomial time and space. 
\end{quote}
\end{abstract}
\section{Introduction}

With the increasing complexity of digital circuits, ensuring functional correctness is a challenging problem. Various techniques for verification have been proposed. While simulation and emulation can be carried out very fast, the coverage that can be obtained is very low. Only formal verification approaches can ensure 100\% correctness (see e.g.~\cite{Dre:2018}). Several automatic proof techniques have been proposed for the bit-level, like e.g.~BDDs or SAT solvers, and for the word-level, like WLDDs or SMT solvers. For an overview see \cite{DS:2001}. In the following, only BDDs are considered, while many of the arguments can directly be transferred to other proof techniques as well. 

Based on formal proof techniques, correctness can be ensured, but the techniques have exponential worst case behavior and for this are known to not scale well in general. Even in cases where the final result can efficiently be represented, intermediate results might be too large to be handled. 
\begin{example}\label{ex:miter}
Consider the general formulation of circuit equivalence based on a miter circuit in Figure \ref{fi:miter}. If the two circuits {\em Circuit 1} and {\em Circuit 2} are functionally equivalent, the output signal {\em out} becomes $0$. This can be represented trivially e.g.~by a BDD. But dependent on the circuits, the construction of the intermediate BDDs might fail.  
\end{example} 
\begin{figure}[t]
\begin{center}
\includegraphics[scale=1.0]{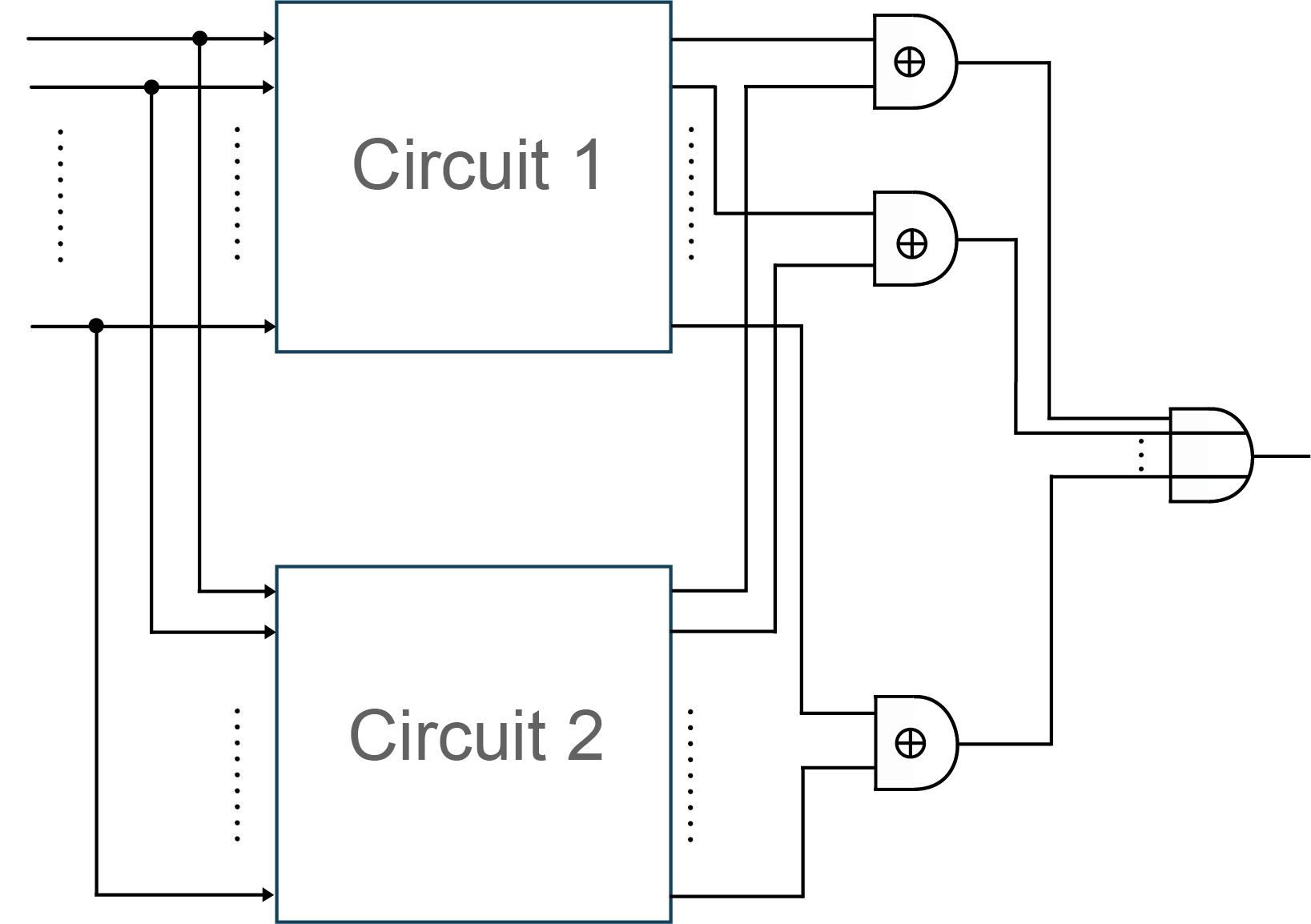}
\end{center}
\caption{Miter circuit}\label{fi:miter}       
\end{figure}
This has already been observed in early applications of BDDs and as a consequence approaches have been presented to reduce the peak size during the construction \cite{SBSV:93,JNC+:95,DG:99c}. These techniques were either heuristic in nature or had prohibitive large run times. Efficient construction could not be guaranteed in general. 

Recently it has been proven in \cite{Dre:2021} that for different kind of adder circuits a complete formal verification based on BDDs can be carried out in polynomial time and space. This not only holds for the representation of the outputs, but for the complete construction.

While the study in \cite{Dre:2021} was restricted to adders, in the following it is shown that for restricted types of circuits polynomial verification based on BDDs can be carried out. This contains tree-like circuits and circuits derived by a 1-to-1 mapping from BDDs. The later ones have intensively been studied and are known for their good testability properties (see e.g.\cite{Bec:92,DSF:2004}). 

\section{Notation and Definition}\label{se:notdef}

Let $f: {\bf B}^n \rightarrow {\bf B}^m$ be a Boolean function  over variable set $X_n = \{x_1, \ldots, x_n \}$. 

\subsection{Circuit}

A {\em circuit} is a {\em Directed Acyclic Graph} (DAG) $G(V,E)$ with vertex set $V$ and edges $E$. It represents a Boolean function $f: {\bf B}^n \rightarrow {\bf B}^m$ with $n$ variables that are associated with the {\em Primary Inputs} (PIs) and $m$ output signals associated with the {\em Primary Outputs} (POs). The internal nodes are associated with basic Boolean functions, like AND, OR, NAND, NOR or INV (=inverter). The edges represent the signals connecting the nodes, PIs and POs. 

\begin{figure}[t]
\begin{center}
\includegraphics[scale=1.0]{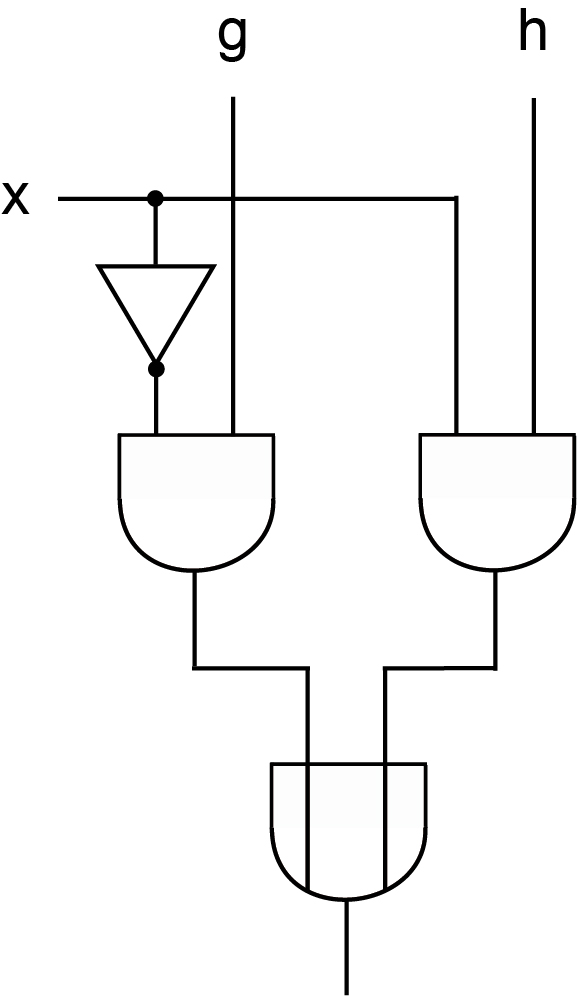}
\end{center}
\caption{MUX representation by basic gates}\label{fi:mux_gate}       
\end{figure}

\begin{example}
The circuit for function $f=\overline{x}\cdot g + x \cdot h$ is shown in Figure \ref{fi:mux_gate}. This circuit is also denoted as multiplexer -- MUX for short -- in the following. 
\end{example}

\subsection{Binary Decision Diagrams}

Reduced ordered Binary Decision Diagrams (BDDs) \cite{Bry:86,DB:98b} are DAGs where a Shannon decomposition
$$f=\overline{x}_if_{\overline{x}_i}+x_if_{x_i} (1 \leq i \leq n)$$ is carried out in each node. The {\em size} of a BDD is the number of non-terminal nodes. 

An important property of BDDs is that the synthesis operations, like AND, OR or composition, can be carried out in polynomial time and space using the {\em apply}-operation \cite{Bry:85}. This can also be described by the operator {\em if-then-else} (ite) \cite{BRB:90}\footnote{Notice that in the following for the discussion and the proofs BDDs without CEs are considered.}, since {\em ite} allows for an easy and elegant implementation in a BDD package. A sketch of the algorithm is shown in Figure \ref{fi:ite}, where  $T$ (=then) and $E$ (=else) denote the high- and low-successors, respectively, and e.g.~$F_{x_i}$ ($F_{\overline{x}_i}$) is the cofactor to $1$ ($0$) with respect to variable $i$:

\begin{figure}
\begin{array}[t]{lllrr}
ITE (F,G,H) \{ & \\ 
	& \text{if} \, (\text{terminal case}) \{ \\
     &  \hspace{1cm} \text{return} \, (\text{result of terminal case}); \\
     &  \} \, \text{else} \, \{ \\
       &  \text{if} \, ((F,G,H) \in \text{computed-table})  \\ 
     &	\hspace{1cm} \text{return \text{computed\_table}} \, (F,G,H); \\
     &  \} \, \text{else} \, \{ \\ 
	& 	\hspace{1cm} \text{let} \, x_i \, \text{be the top variable of} \, (F,G,H); \\
	&	\hspace{1cm} T = ITE(F_{x_i},G_{x_i},H_{x_i}); \\ 
	&	\hspace{1cm} E = ITE(F_{\overline{x}_i},G_{\overline{x}_i},H_{\overline{x}_i}); \\
	&	\hspace{1cm} \text{if} \, (T = E) \, \text{return} \, T ; \\
	&	\hspace{1cm} R = \text{find\_or\_add\_unique\_table} (v, T, E); \\
	&   \hspace{1cm} \text{insert\_computed\_table} \, (F,G,H,R); \\
	&   \hspace{1cm} \text{return} \, R; \\
     & \}  \\
\}
\end{array}
\caption{ITE algorithm}\label{fi:ite}
\end{figure}

The {\em ite}-operator has a polynomial worst case behavior, i.e.~for graphs $F$, $G$ and $H$ the result is bound by $O(|F| \cdot |G| \cdot |H|)$. This bound holds under the assumption of an optimal hashing in $O(1)$. But also in the case of a worst case behavior of the hashing function, {\em ite} remains polynomial (see \cite{DS:2001b}).

Simple examples for terminal cases that might occur are as follows: 
\begin{itemize}
  \item if (F == 1 || G == H) return G;    
  \item if (F == 0) return H;
\end{itemize}
In software implementations of BDDs, typically more complex cases are considered to allow for an early termination of the recursive calls (see e.g.~\cite{Som:2001a,Som:2001}). 

\subsection{Symbolic Simulation}

To build the BDDs for the output signals of a circuit, the circuit is traversed in a topological order starting from the inputs. For the inputs signals the corresponding BDDs are initially generated. Then, for each gate in the circuit the corresponding synthesis operation based on {\em ite} is carried out. This process is called {\em symbolic simulation} in the following. 

\begin{example}
The symbolic simulation for a circuit consisting of a single AND gate is shown in Figure \ref{fi:symbsym}.
\end{example}

\begin{figure}[t]
\begin{center}
\includegraphics[scale=1.0]{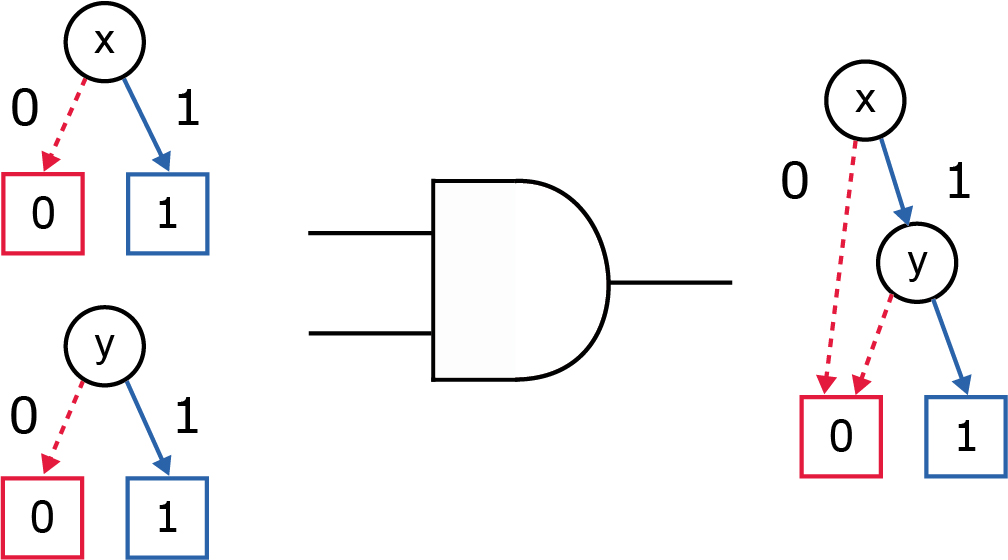}
\end{center}
\caption{Symbolic simulation for AND gate}\label{fi:symbsym}       
\end{figure}

\section{Polynomial Verification}\label{se:poly_ver}

While for several types of circuits it is known that the BDD representing the output has a polynomial size, it is not obvious that during the construction of the output BDD (significantly) larger intermediate results are computed (see Example \ref{ex:miter}). In the following we consider different circuit structures and show that under some constraints a polynomial symbolic simulation can be ensured. 

\subsection{Tree-like Circuits}

For tree-like circuits, i.e.~circuits without fanouts, it is well known that the BDDs representing the output function have linear size. This means that for each input variable only a single node is needed, if the variable ordering results from a {\em Depth First Search} (DFS) traversal of the circuit.\footnote{Notice that this only holds for basic gates, like AND, OR, NAND, NOR, while e.g.~EXOR has to be handled differently (see e.g.~\cite{DBJ:98}).}  This is typically proven by induction (see e.g. Lemma 5.3 (Section 5.4.1) in \cite{DB:98b}). But this again only considers the BDD for the output of the circuit, but not during the construction. For this, in the following we will use a different argument to obtain the same result. But this argument can later be generalized to also reason about circuits that do not have a tree-like structure.

For the AND-gate the value $0$ ($1$) is a {\em controlling value} (cv) ({\em non-controlling value} (ncv)), since the output is (not) determined, if one of the inputs assume this value. In a similar way these values can be described for other Boolean gates, like NAND, OR or NOR. 

When constructing BDDs based on the {\em ite}-operator introduced above, a {\em cv} directly leads to a terminal case and by this ends the recursion. This implies that not only the final BDDs have a linear size in the number of input variables, but also during the construction it can never occur that additional nodes are generated. 

The same argument can be applied to general circuits.
\begin{theorem}\label{th:top_variable}
Let $g$ be a BDD representing a circuit and $x_i$ a variable that does not occur in $g$. Then $f = x_i \cdot g$ can be constructed in $O(|g|)$, if $x_i$ is the top-variable in the ordering of $f$.
\end{theorem}
\begin{proof}
For the case of $x_i=0$, in the recursive call of {\em ite}, directly a terminal case is reached, since $0$ is the {\em cv} for the AND-gate. On the 2nd branching, the BDD of $g$ has to be traversed, but no additional nodes are created. 
\end{proof}
The same results hold, if the AND-operation is substituted by OR, NAND or NOR or for the complemented variable $\overline{x}_i$. Also the same argument can be applied if we consider gates with more than two inputs, e.g.~a 4-input AND. Then the gate has to be first decomposed into 2-input gates and then the theorem can be applied. 

\subsection{General Circuits}

For arbitrary circuits without any restrictions on the structure, polynomial symbolic simulation can not be ensured. E.g.~in the case of the multiplier function in contrast each BDD construction will lead to an exponential blow-up, since the final BDD has an exponential size (see \cite{Bry:91}). But, if the sizes of the BDDs representing the internal signals of the circuit are polynomially bound, the complete construction of the output BDDs can be carried out efficiently:
\begin{theorem}\label{th:poly_circ}
Let $C$ be a circuit over $n$ input variables. If the BDD size of each internal signal of the circuit has a size polynomial in $n$, then the complete symbolic simulation of the circuit can be carried out in polynomial time and space.
\end{theorem}
\begin{proof}
This directly follows from the observation that {\em ite} has a polynomial worst case behavior and the fact that each internal signal has a polynomial BDD representation. 
\end{proof}
The theorem can be easily be relaxed to not consider all internal signals.
\begin{remark}\label{re:const}
To ensure polynomial worst case behavior not all internal signals have to be considered, while it is sufficient to have a constant number $c$ of gates between these signals. The number $c$ should not be too large, since this would result in a still polynomial but for practical reasons too large estimate.
\end{remark}
This property will be used in Section \ref{sub:bdd_circuit} and was also (implicitly) applied in \cite{Dre:2021} for the special case of adder circuits, i.e.~the {\em Conditional Sum Adder} (CoSA) particularly. 

\subsection{BDD-Circuits}\label{sub:bdd_circuit}

A BDD-circuit results from a BDD by substituting each internal node of the BDD by a MUX (see Figure \ref{fi:mux_gate}). 

\begin{example}
For the OR-function $f_{OR}=x_1 + x_2$ the BDD is shown on the left hand side of Figure \ref{fi:bdd_circuit}. The resulting circuit is shown on the right hand side, where each internal node is substituted by a MUX cell. Notice that the circuit is drawn upside down compared to the BDD in the common way that the outputs are at the lower end of the figure. 
\end{example}
Since this is a 1-to-1 mapping, it is obvious that the BDD for each of the signals in the circuit have polynomial size. Thus, Theorem \ref{th:poly_circ} can be applied. When we also want to argue over the internal signals of the MUX, the way the MUX is represented has to be considered as well. If we assume the standard realization from Figure \ref{fi:mux_gate} and following Remark \ref{re:const}, it can be observed:
\begin{enumerate}
\item The side input $x_i$ only consists of one node. 
\item The negation of this input also consists of one node only. 
\item Due to the ordering restriction of the original BDD, the function represented at the inputs (beside the selection signal) of the MUX are independent of $x_i$ (see also Theorem \ref{th:top_variable}). 
\item Thus, the internal signals at the outputs of the AND-gates have polynomial sizes. 
\item Due to the polynomial worst-case behavior of the {\em ite}-operator, also the output of the OR-gate has polynomial size.
\end{enumerate}
In summary, we obtain: 
\begin{theorem}
For BDD-circuits with MUX using the standard representation the complete symbolic simulation of the circuit can be carried out in polynomial time and space.
\end{theorem}
\begin{figure}[t]
\begin{center}
\includegraphics[scale=0.75]{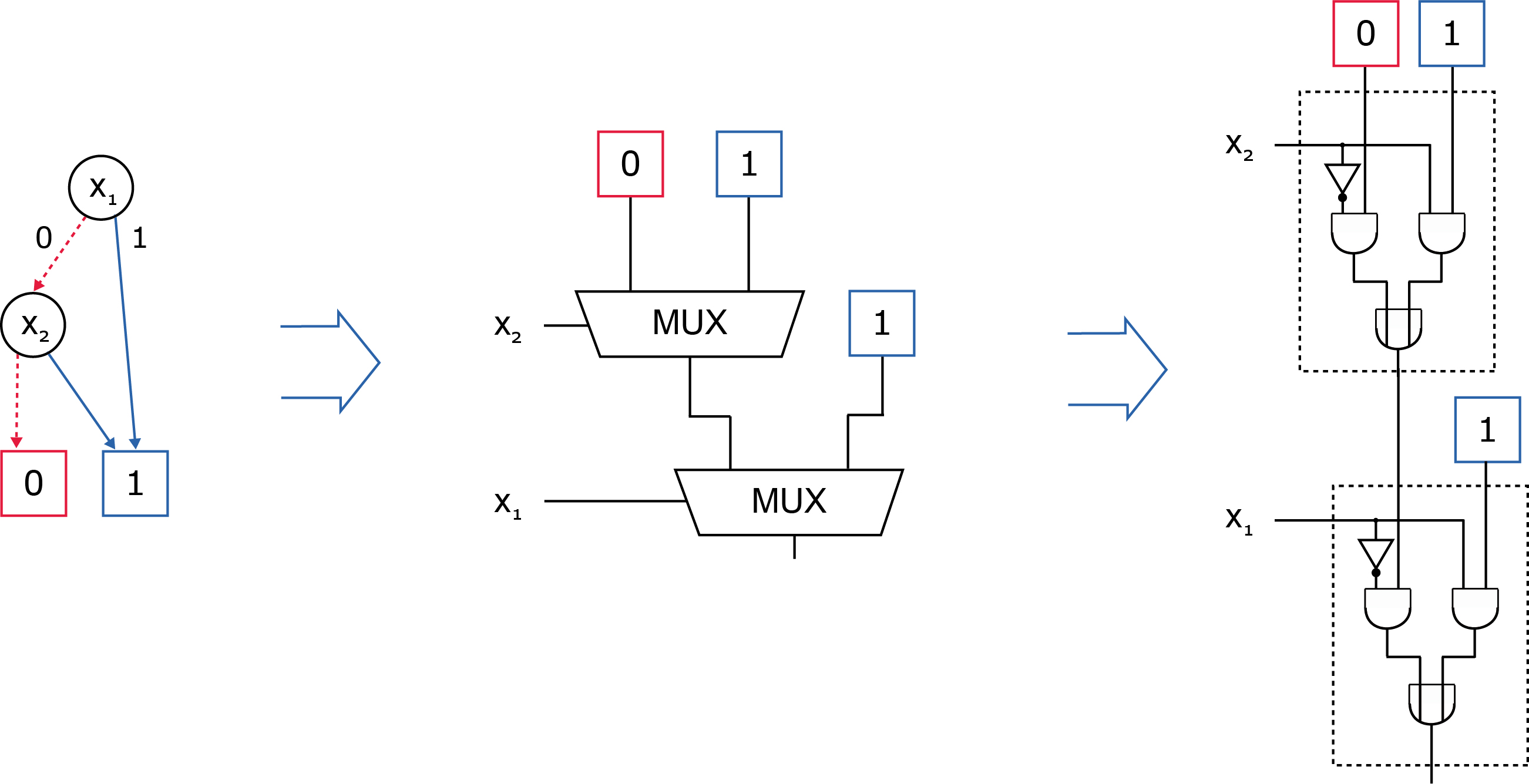}
\end{center}
\caption{BDD and resulting circuit}\label{fi:bdd_circuit}       
\end{figure}

\section{Conclusion}\label{sec:concl}

Ensuring verification of designed circuits is a central aspect of successful chip design in the future. The concept of polynomial circuit verification has only been demonstrated before for adder circuits, while here it has been shown to be applicable to more general Boolean functions. This includes functions that can be represented by tree-like circuits. Furthermore, it has been shown that circuits derived from BDDs by a 1-to-1 mapping can efficiently be formally verified. 

Future work will go into two directions: On the one hand side the class of functions that can be polynomially verified using BDDs should be extended. Besides this, alternative proof concept can be investigated, like it has been done successfully using BMDs for multipliers (see \cite{KMB+:97}). The investigation becomes more difficult here, since polynomial synthesis operations, like {\em ite} for BDDs, do not exist in general in this case. 

\section*{Acknowledment}
Parts of this work have been supported by DFG within the Reinhart Koselleck Project {\em PolyVer: Polynomial Verification  of Electronic Circuits} (DR 287/36-1). Furthermore, the author likes to thank Alireza Mahzoon for helpful comments and discussions. 

\bibliographystyle{latex8}
\bibliography{lit_bank,fey_loc,grosse_loc,lit_bank_ext}

\end{document}